\documentclass[12pt,reqno,a4paper]{amsart}
\usepackage{amsmath,amssymb,dsfont,graphicx,cite,subfigure,enumerate,rotating,mathrsfs}
\usepackage{epic,eepic}
\usepackage[colorlinks=false]{hyperref}
\def\beq{\begin{equation}}
\def\eeq{\end{equation}}
\def\bea{\begin{eqnarray}}
\def\eea{\end{eqnarray}}

\newtheorem{theorem}{Theorem}

\newtheorem{definition}{Definition}
\newtheorem{corollary}{Corollary}

\makeatletter
\expandafter\let\expandafter
\reset@font\csname reset@font\endcsname
\def\subeqnarray{\arraycolsep1pt
   \def\@eqnnum\stepcounter##1{\stepcounter{subequation}
       {\reset@font\rm(\theequation\alph{subequation})}}
\jot5mm     \eqnarray}

\makeatother

\def\epsilon{\varepsilon}

\def\t{\widetilde}
\def\nn{\nonumber}
\def\h{\widehat}
\def\endpf{\hfill$\square$\medskip}

\newcommand{\CT}{{\mathscr{T}}}
\newbox\meibox
\def\placeunder#1#2#3#4{\setbox\meibox%
\vbox{\hbox{\hskip#4$\hphantom{#2}$}\hbox{$\hphantom{#1}$}}%
\vtop{\baselineskip=0pt\lineskiplimit=\baselineskip%
\lineskip=#3\hbox to \wd\meibox{\hfil\hskip#4$#2$\hfil}%
\hbox to \wd\meibox{\hfil$#1$\hfil}}}

\begin{document}
\title[Spherical geometry and integrable systems]
{Spherical geometry and integrable systems}

\author{Matteo Petrera \and Yuri B. Suris }

\thanks{E-mail: {\tt  petrera@math.tu-berlin.de, suris@math.tu-berlin.de}}

\maketitle

\begin{center}
{\footnotesize{
Institut f\"ur Mathematik, MA 7-2\\
Technische Universit\"at Berlin, Str. des 17. Juni 136,
10623 Berlin, Germany
}}
\end{center}

\begin{abstract}
We prove that the cosine law for spherical triangles and spherical tetrahedra defines integrable systems, both in the sense of multidimensional consistency and in the sense of dynamical systems.
\end{abstract}


\section{Introduction}
\label{sect: intro}

Nowadays, it is a well accepted fact that many fundamental notions and results of differential geometry are closely related with the theory of integrable systems. Moreover, the recent development of discrete differential geometry \cite{DDG} led to uncovering and identifying the true roots of such relations in the most fundamental incidence theorems of projective geometry and its metric subgeometries. On the other hand, this development led also to a clarification of the notion of integrability, which is conceptually much more clear at the level of discrete systems. In particular, the notion of multidimensional consistency as integrability crystallized in the course of this development (see \cite{DDG}).

In the present paper, we offer a new evidence in favor of the thesis that the very basic structures of geometry are integrable in their nature. We show that the cosine laws for spherical simplices define integrable systems in several senses. 

Our main results are as follows:

\smallskip

\begin{itemize}

\item We show that the map sending planar angles of a spherical triangle to its side lengths is integrable in the sense of multidimensional consistency (Sect. \ref{sect: dDarboux}). Moreover, we identify this map with the symmetric reduction of the well known discrete Darboux system, which appeared for the first time in the context of discrete conjugate nets. We also discuss the relation of this result with a geometric realization of tetrahedral angle structures, due to Feng Luo \cite{Luo2}.

\smallskip

\item We show that the same map, interpreted as a dynamical system, is completely integrable in the Liouville-Arnold sense (Sect. \ref{sect: dEuler}). This map can be considered as a non-rational discretization of the famous Euler top, which is a basic example of integrable system. Moreover, we identify the second iterate of this map as a birational discretization of the Euler top, introduced by Hirota and Kimura \cite{HK} and studied in detail in our previous publication \cite{PS}. This second iterate is shown to admit a nice geometric interpretation as a particularly simple and nice map between spherical triangles, which we call {\it{switch}} (the angles of the switched triangle are set to be equal to the sides of the original one). A closely related involutive map between spherical triangles was studied in a paper by Jonas \cite{J},  which can be thus considered as an early instance of appearance of integrable systems in elementary geometry.

\smallskip

\item We show that the map sending dihedral angles of a spherical tetrahedron to its side lengths is a further example of an integrable dynamical system (Sect. \ref{sect: dcoupled Euler}). For this map, we establish conserved quantities, which turn out to be a re-interpretation of the sine laws, and an invariant volume form. The latter is shown to be closely related, via the famous Schl\"afli formula, to the determinant of the Hesse matrix of the volume of the spherical tetrahedron as a function of dihedral angles.
\end{itemize}
For convenience of the reader, we recall (in Sect. \ref{sect: simplices}) the standard derivation of the spherical cosine and sine laws via Gram matrices.

\section{Cosine law for spherical simplices}
\label{sect: simplices}

The maps we are dealing with in the present paper have their origin in the cosine law for spherical simplices. We  understand a spherical $(n-1)$-simplex as an intersection of $n$ hemispheres in $\mathbb S^{n-1}$ in general position. Let $v_1,\ldots,v_n\in\mathbb S^{n-1}$ be the vertices of a spherical $(n-1)$-simplex, and let $v_1^*,\ldots,v_n^*\in\mathbb S^{n-1}$ be the vertices of the polar simplex, defined by the conditions
\[
\langle v_i^*,v_j\rangle=0\quad {\rm if}\;  i\neq j,\qquad \langle v_i^*,v_i\rangle>0.
\]
The edge lengths $\ell_{ij}\in(0,\pi)$ and the dihedral angles $\alpha_{ij}\in(0,\pi)$ are given by
\[
\cos\ell_{ij}=\langle v_i,v_j\rangle,\qquad \cos\alpha_{ij}=-\langle v_i^*,v_j^*\rangle.
\]
Each set of data $\{\ell_{ij}\}$ or $\{\alpha_{ij}\}$ defines a spherical simplex uniquely, and the spherical cosine law gives a transition from one set of data to another. The most direct way to derive the cosine law is through the so called Gram matrices
\[
G=(-\cos\alpha_{ij})_{i,j=1}^n,\qquad G'=(\cos\ell_{ij})_{i,j=1}^n.
\]
According to Lemma 1.2 in \cite{Luo1}, Gram matrices of spherical simplices are characterized by the following properties: they are positive definite symmetric matrices with diagonal entries equal to 1. Introducing the matrices $V$ and $W$ whose columns are $v_i$, resp. $v_i^*$, we can write:
\[
W^{\rm T}W=G, \qquad V^{\rm T}V=G', \qquad V^{\rm T}W=D,
\]
where $D$ is a diagonal matrix with positive entries. There follows:
\beq\label{Grams rel}
G'=DG^{-1}D,
\eeq
from which the {\em cosine laws} follow immediately. They can be put as
\beq\label{cos gen}
\cos\ell_{ij}=\frac{g_{ij}}{\sqrt{g_{ii}g_{jj}}},\qquad \cos\alpha_{ij}=\frac{g'_{ij}}{\sqrt{g'_{ii}g'_{jj}}},
\eeq
where $g_{ij}$, $g'_{ij}$ are the cofactors of the Gram matrices $G$, $G'$, respectively, see, for instance, \cite{V}. We now consider in more detail the particular cases $n=3,4$ of this construction.

\subsection{Spherical triangles}

For a spherical triangle, $\alpha_{ij}$ is just the plane (inner) angle opposite to the edge $\ell_{ij}$,  see Fig.~\ref{Fig: triangle}.

\begin{figure}[h!]
\includegraphics[height=5.5cm]{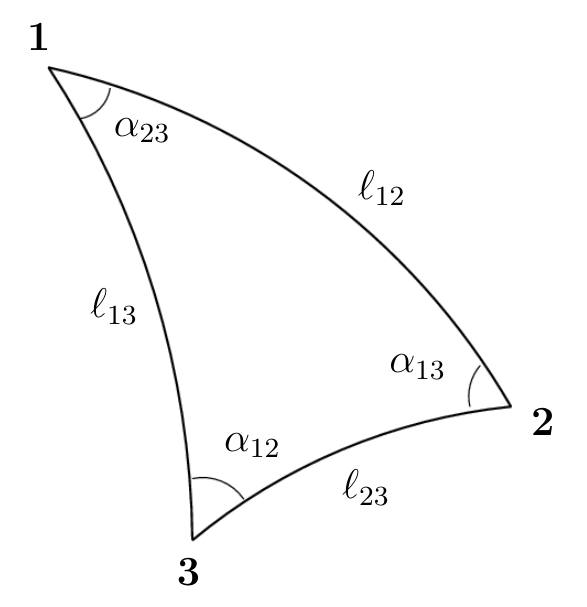}
\caption{A spherical triangle.}
\label{Fig: triangle}
\end{figure}

The angles of a spherical triangle belong to the tetrahedron-shaped domain
\beq
\CT= \{\alpha=(\alpha_{12},\alpha_{13},\alpha_{23}) \in (0,\pi)^3 \, : \, \alpha_{12}+\alpha_{13}+ \alpha_{23} > \pi,  \;-\alpha_{ij} + \alpha_{ik} + \alpha_{jk} < \pi\}. \nn
\eeq
Here and below we use $(i,j,k)$ to denote any permutation of (1,2,3). The sides belong to the domain
\beq
\CT^*= \{\ell=(\ell_{12},\ell_{13},\ell_{23}) \in (0,\pi)^3 \, : \, \ell_{12}+\ell_{13}+ \ell_{23} < 2 \pi,  \;
-\ell_{ij} + \ell_{ik} + \ell_{jk}>0\},
\nn
\eeq
which can be obtained from $\CT$ via the reflection $I:\CT\leftrightarrow\CT^*$, $I(\alpha)=\pi-\alpha$.
The space of spherical triangles (up to orthogonal transformations of $\mathbb R^3$) is parametrized by either of $\CT$ or $\CT^*$: a spherical triangle exists if and only if $\alpha \in \CT$, resp. if and only if $\ell\in\CT^*$.

The {\em cosine laws} for spherical triangles, as given in (\ref{cos gen}), read:
\bea
&& \cos\ell_{ij}=  \frac{\cos\alpha_{ij} +\cos\alpha_{ik}\cos\alpha_{jk}}
{ \sin \alpha_{ik} \sin \alpha_{jk}}, \label{c1} \\
&&\cos\alpha_{ij}=  \frac{\cos\ell_{ij} -\cos\ell_{ik}\cos\ell_{jk}}
{ \sin \ell_{ik} \sin \ell_{jk}}. \label{c2}
\eea
We will consider formulas (\ref{c1}--\ref{c2}) as defining the mutually inverse maps
\beq \nn
\begin{array}{rcl}
\Phi: \CT & \to     & \CT^*, \\
\alpha    & \mapsto & \ell,
\end{array}\qquad
\begin{array}{rcl}
\Phi^{-1}: \CT^* & \to     & \CT, \\
 \ell   & \mapsto & \alpha.
\end{array}
\eeq
We note that $\Phi^{-1}=I\circ\Phi\circ I$.

The cosine law can be supplemented by the {\em sine law} for spherical triangles:
\beq\label{si}
\frac{\sin\ell_{12}}{\sin\alpha_{12}}=\frac{\sin\ell_{13}}{\sin\alpha_{13}}=\frac{\sin\ell_{23}}{\sin\alpha_{23}}=
\frac{\sqrt{d}}{\gamma}=\frac{\gamma'}{\sqrt{d'}},
\eeq
where
\bea 
\nn
 d & = & \det G\ =\ 1-\cos^2\alpha_{12}-\cos^2\alpha_{13}-\cos^2\alpha_{23}-2\cos\alpha_{12}\cos\alpha_{13}\cos\alpha_{23},\\
 d' & = & \det G'\ =\ 1-\cos^2\ell_{12}-\cos^2\ell_{13}-\cos^2\ell_{23}+2\cos\ell_{12}\cos\ell_{13}\cos\ell_{23}, \nn
\eea
and
$$
\gamma= \sin\alpha_{12}\sin\alpha_{13}\sin\alpha_{23},\qquad 
\gamma' = \sin\ell_{12}\sin\ell_{13}\sin\ell_{23}.\nn
$$
The sine law is derived by considering principal $2\times 2$ minors of the relation (\ref{Grams rel}).

Upon the change of variables
\[
x_{ij}=\cos\alpha_{ij},\quad y_{ij}=\cos\ell_{ij},
\]
the cosine law for spherical triangles takes the form
\bea
&& y_{ij}=  \frac{x_{ij}+x_{ik}x_{jk}}
{\sqrt{1-x_{ik}^2}\sqrt{1-x_{jk}^2}}, \label{c1 rat} \\
&&x_{ij}=  \frac{y_{ij}-y_{ik}y_{jk}}
{ \sqrt{1-y_{ik}^2}\sqrt{1-y_{jk}^2}}.\label{c2 rat}
\eea
We denote these maps by
$$
\begin{array}{rcl}
\phi: \tau & \to     & \tau^*, \\
x    & \mapsto & y,
\end{array}\qquad
\begin{array}{rcl}
\phi^{-1}: \tau^* & \to     & \tau, \\
 y   & \mapsto & x,
\end{array}
$$
and call them the algebraic forms of $\Phi, \Phi^{-1}$, respectively. They are defined on the corresponding subsets $\tau,\tau^*$ of $[-1,1]^3$, and are related by $\phi^{-1}=i\circ \phi\circ i$, $i(x)=-x$. The sine law in these new variables reads:
\beq
\frac{1-y_{12}^2}{1-x_{12}^2} = \frac{1-y_{13}^2}{1-x_{13}^2}= \frac{1-y_{23}^2}{1-x_{23}^2}=
\frac{d}{\gamma^2}=\frac{(\gamma')^2}{d'},
\label{si rat}
\eeq
where
\beq\label{po rat 1}
d=1-x_{12}^2-x_{13}^2-x_{23}^2-2x_{12}x_{13}x_{23},\quad
d'=1-y_{12}^2-y_{13}^2-y_{23}^2+2y_{12}y_{13}y_{23},
\eeq
and
$$
\gamma^2=(1-x_{12}^2)(1-x_{13}^2)(1-x_{23}^2), \qquad  (\gamma')^2=(1-y_{12}^2)(1-y_{13}^2)(1-y_{23}^2).
$$
For further reference, we mention that the sine law can be represented in either of the following forms:
\beq \label{eq: sine law prod 1}
 d= (1-y_{ij}^2)(1-x_{ik}^2)(1-x_{jk}^2), \qquad
 d' = (1-x_{ij}^2)(1-y_{ik}^2)(1-y_{jk}^2). 
\eeq

\subsection{Spherical tetrahedra}
\label{subsect tetr}

For a spherical tetrahedron, see Fig.~\ref{Fig: sph tetr}, the {\it{cosine law}} is expressed by the formula (\ref{cos gen}) with
\bea
g_{ij} & = & x_{ij} + x_{ik}x_{jk} + x_{im}x_{jm} + x_{km}(x_{ik}x_{jm}+x_{im}x_{jk}-x_{ij}x_{km}),
\label{sph tetr c1}\\
g_{ii} & = & 1-(x_{jk}^2+x_{jm}^2+x_{km}^2)-2x_{jk}x_{jm}x_{km},
\label{sph tetr c2}
\eea
where we adopt the notation $x_{ij}=\cos\alpha_{ij}$ and $(i,j,k,m)$ stands for any permutation of $(1,2,3,4)$.

\setlength{\unitlength}{0.8cm}
\begin{figure}[http]
\includegraphics[height=9cm]{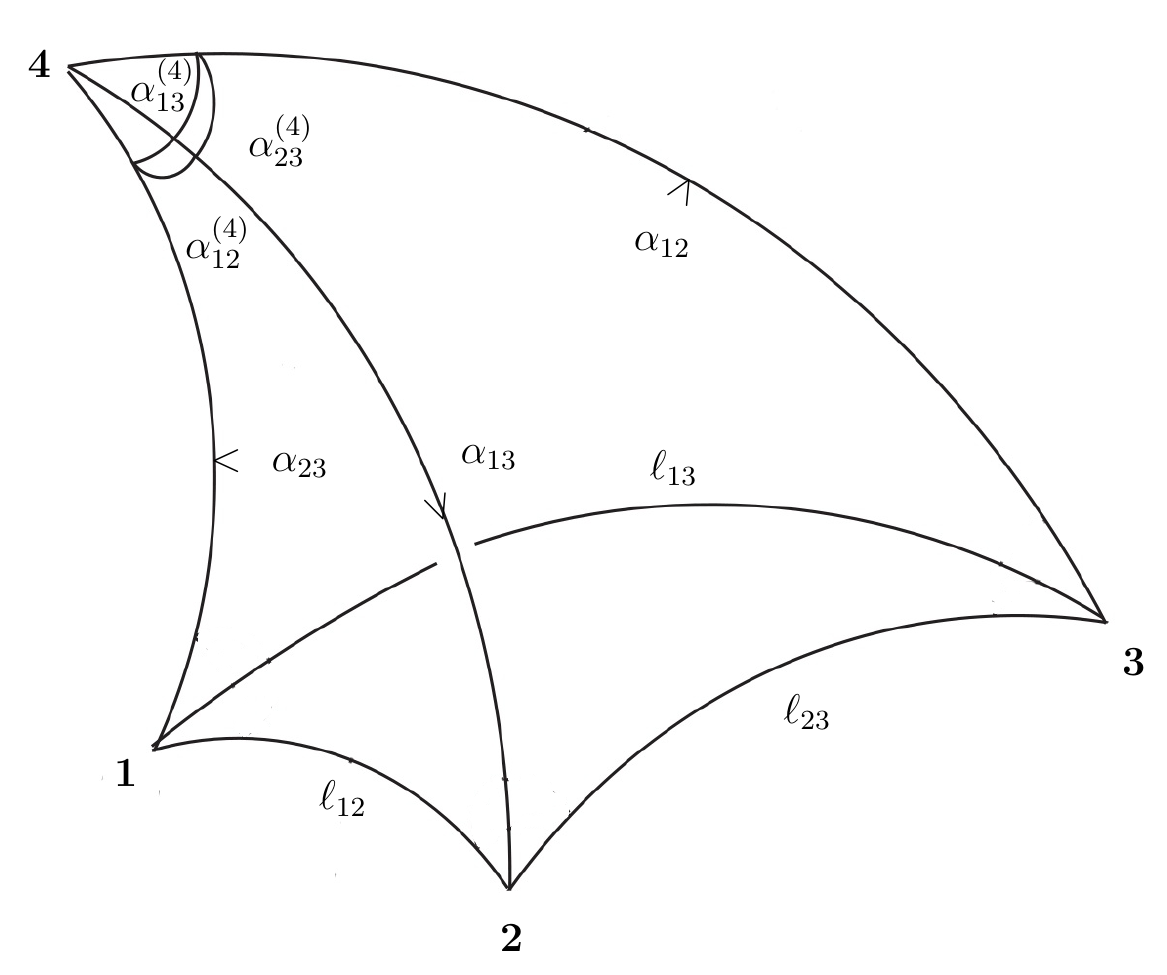}
\caption{A spherical tetrahedron}
\label{Fig: sph tetr}
\end{figure}

These formulae are supplemented (see \cite{DMP,KMP}) by the {\it{sine law}},  which can be put as
\bea
\frac{\sin \ell_{12} \sin \ell_{34}}{\sin \alpha_{12} \sin \alpha_{34}} =
\frac{\sin \ell_{13} \sin \ell_{24}}{\sin \alpha_{13} \sin \alpha_{24}} =
\frac{\sin \ell_{14} \sin \ell_{23}}{\sin \alpha_{14} \sin \alpha_{23}} =
\frac{d}{\gamma}=\frac{\gamma'}{d'}, \label{67}
\eea
and by the formulas
\bea
\frac{\cos \ell_{12} \cos \ell_{34}- \cos \ell_{13} \cos \ell_{24}}
{\cos \alpha_{12} \cos  \alpha_{34}- \cos  \alpha_{13} \cos  \alpha_{24}} &=&
\frac{\cos \ell_{13} \cos \ell_{24}- \cos \ell_{14} \cos \ell_{23}}
{\cos \alpha_{13} \cos  \alpha_{24}- \cos  \alpha_{14} \cos  \alpha_{23}} \nn \\
&=& \frac{\cos \ell_{14} \cos \ell_{23}- \cos \ell_{12} \cos \ell_{34}}
{\cos \alpha_{14} \cos  \alpha_{23}- \cos  \alpha_{12} \cos  \alpha_{34}} =
\frac{d}{\gamma}=\frac{\gamma'}{d'}, \label{68}
\eea
where $d=\det G$, $d'=\det G'$, and $\gamma= \sqrt{g_{11}g_{22}g_{33}g_{44}}$, $\gamma'= \sqrt{g'_{11}g'_{22}g'_{33}g'_{44}}$.

\section{Spherical triangles and symmetric discrete Darboux system}
\label{sect: dDarboux}

We now consider a somewhat fancy combinatorial interpretation of formulas (\ref{c1 rat}) for the cosine law for spherical triangles. In this interpretation, we combinatorially assign the three inner angles of a spherical triangle and its three side lengths to the six 2-faces of an elementary cube, in such a manner that each angle and an opposite side are assigned to two opposite faces of the cube. More precisely, we assign the quantities $x_{jk}$ to the three faces of a 3-dimensional cube parallel to the coordinate plane $jk$, and the quantities $y_{jk}=T_ix_{jk}$ to the three opposite faces, see Fig. \ref{Fig:cube eq faces}. Here, $T_i$ stands for the unit shift in the $i$-th coordinate direction.

\begin{figure}[htbp]
\begin{center}
\includegraphics[height=4.5cm]{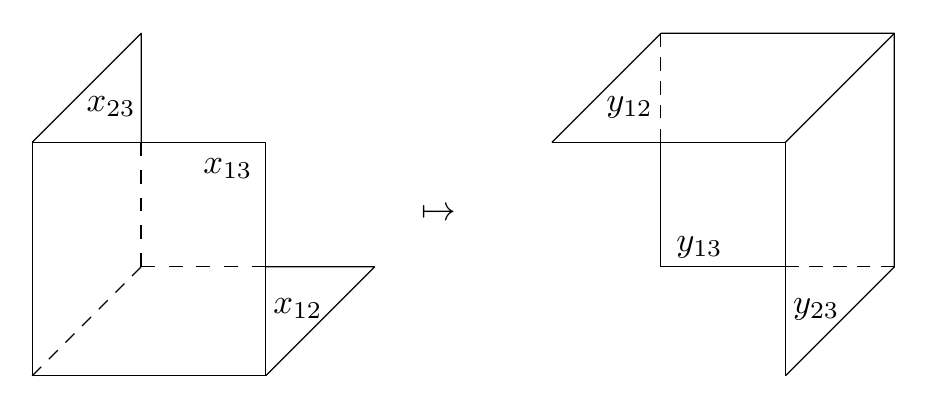}
\caption{A map on an elementary 3D cube with fields assigned to 2-faces. 
}
\label{Fig:cube eq faces}
\end{center}
\end{figure}

Thus, what we consider now is, in the terminology of \cite{DDG}, a 3D system with fields assigned to elementary squares, given by the formulas
$$
(T_3x_{12},T_2x_{13},T_1x_{23})=\phi\left(x_{12},x_{13},x_{23}\right),
$$
or, explicitly,
\beq\label{eq: dDarboux sym}
\left\{\begin{array}{l}
T_3x_{12}=\dfrac{x_{12}+x_{13}x_{23}}{\sqrt{1-x_{13}^2}\sqrt{1-x_{23}^2}},\\
T_2x_{13}=\dfrac{x_{13}+x_{12}x_{23}}{\sqrt{1-x_{12}^2}\sqrt{1-x_{23}^2}},\\
T_1x_{23}=\dfrac{x_{23}+x_{12}x_{13}}{\sqrt{1-x_{12}^2}\sqrt{1-x_{13}^2}}.
\end{array}\right.
\eeq
This 3D system will be called the {\em symmetric discrete Darboux system}. It can be seen as the special case (symmetric reduction) $x_{ij}=x_{ji}$ of the general {\em discrete Darboux system}, given by
\beq\label{eq: dDarboux}
T_kx_{ij}=\frac{x_{ij}+x_{ik}x_{kj}}{\sqrt{1-x_{ik}x_{ki}}\sqrt{1-x_{kj}x_{jk}}}.\\
\eeq
System (\ref{eq: dDarboux}) is well known in the theory of discrete integrable systems of geometric origin. It seems to have appeared for the first time in \cite[eq. (7.20)]{KS} as a gauge version of a simpler but less symmetric form of the discrete Darboux system,
\beq\label{eq: dDarboux alt}
T_kx_{ij}=\frac{x_{ij}+x_{ik}x_{kj}}{1-x_{kj}x_{jk}},\\
\eeq
derived earlier in \cite{BK}. Either form of the discrete Darboux system describes the so called rotation coefficients of discrete conjugate nets $f:\mathbb Z^3\to\mathbb R\mathbb P^n$, i.e., nets with planar elementary quadrilaterals, see \cite[p. 42]{DDG}. The version of rotation coefficients satisfying system (\ref{eq: dDarboux}) has an advantage of being defined in more local terms, see \cite[p. 95--96]{DDG}. Besides, system (\ref{eq: dDarboux}) admits, unlike system (\ref{eq: dDarboux alt}), the symmetric reduction $x_{ij}=x_{ji}$, and is therefore better suited for the description of the so called symmetric conjugate nets and discrete Egorov nets, see \cite{D}. The symmetric discrete Darboux system is a close relative of the {\em discrete CKP system}, see \cite{S}.

Thus, we have proved that the symmetric discrete Darboux system (\ref{eq: dDarboux sym}) admits an alternative interpretation in terms of spherical triangles. This interpretation leads also to a new proof of integrability of the symmetric discrete Darboux system. Let us recall that integrability of discrete 3D systems is synonymous with their 4D consistency  \cite{DDG} . The meaning of this notion can be described as follows. 

Consider the initial value problem with the data $x_{ij}=x_{ji}$, $i,j=1,2,3,4$, prescribed at six squares adjacent to one common vertex of the 4D cube. Then the application of a 3D map like (\ref{eq: dDarboux sym}) to the four 3D cubes adjacent to this vertex allows us to determine all $T_kx_{ij}$. At the second stage, the map is applied to the other four 3-faces of the 4D cube, with the result being all $T_m(T_kx_{ij})$ computed in two different ways (since each of the corresponding squares is shared by two 3-faces), see Fig.~\ref{Fig: 4D consistency}. Now, 4D consistency of the map means that  $T_m(T_kx_{ij})=T_k(T_mx_{ij})$ for any permutation $i,j,k,m$ of $1,2,3,4$ and for arbitrary initial data.

\begin{figure}[htbp]
\begin{center}
\includegraphics[height=6cm]{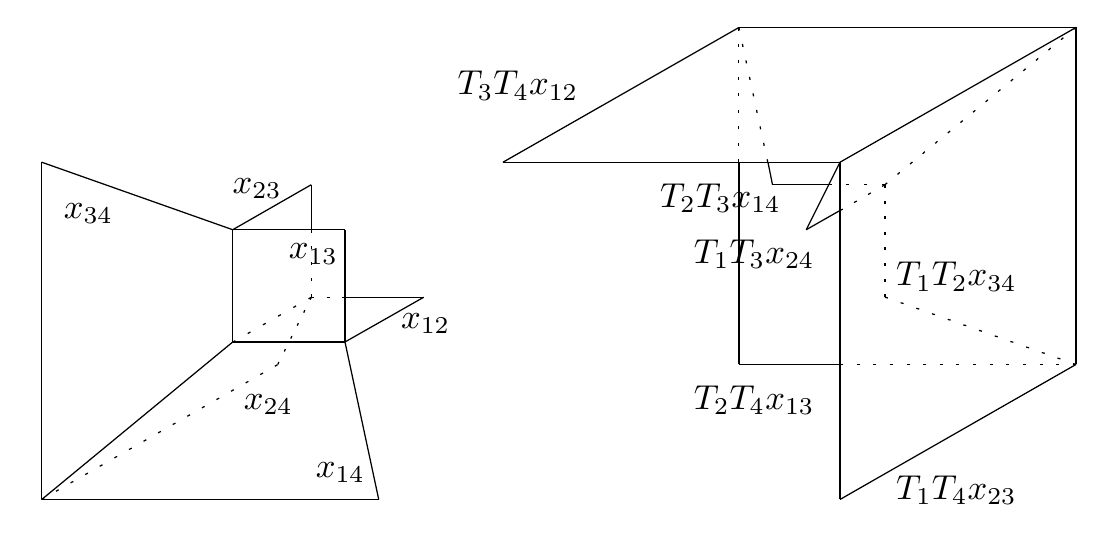}
\caption{Initial data and results of two-stage application of a 4D
consistent map with fields on 2-faces.}
\label{Fig: 4D consistency}
\end{center}
\end{figure}

\begin{theorem}\label{Th: 4D}
Symmetric discrete Darboux system (\ref{eq: dDarboux sym}) is 4D consistent.
\end{theorem}

{\it Proof.} We will show that this is just a corollary of the geometry of spherical tetrahedra. Consider a spherical tetrahedron in $\mathbb S^3$ with the edge lengths $\ell_{ij}$ (connecting the vertices $i$ and $j$) and dihedral angles $\alpha_{ij}$ (at the edges $(km)$, where $k,m$ is the complement of $i,j$ in $1,2,3,4$). One can take either $\alpha_{ij}$ or $\ell_{ij}$ as independent data parametrizing the spherical tetrahedron. The two sets of data are related by the cosine law for spherical tetrahedra, see section \ref{subsect tetr}. What seems to be less known is that this cosine law for spherical tetrahedra can be derived via a two-fold application of the cosine law for the spherical triangles, if the combinatorics of this two-fold application is in agreement with Fig.~\ref{Fig: 4D consistency}. In other words, a spherical tetrahedron combinatorially corresponds to a 4D cube.

\begin{figure}[htbp]
\includegraphics[height=5.5cm]{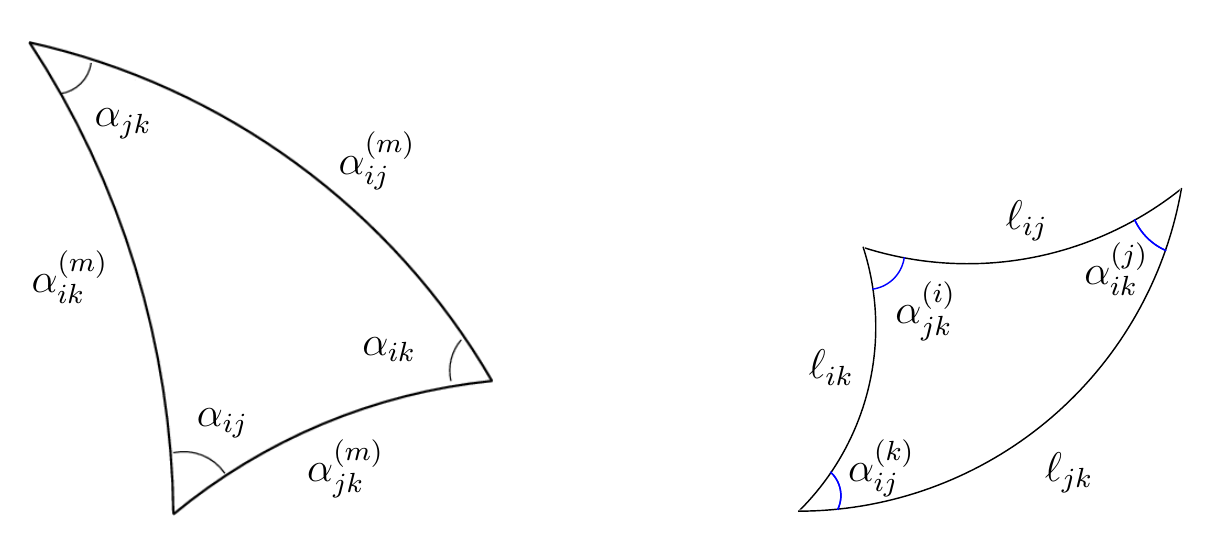}
\caption{Left: Sides and angles in the spherical triangle Lk($m$). Right: Sides and angles in the spherical triangle $\Delta(ijk)$}
\label{Fig5}
\end{figure}

To show this, we start with the dihedral angles $\alpha_{ij}$ as initial data. From these, we can determine all the planar angles $\alpha_{ij}^{(k)}$ (the latter is the notation for the planar angle opposite to the side $(ij)$ in the spherical triangle $\Delta(ijk)$ with the vertices $i,j,k$, see Fig.~\ref{Fig5} (right)). Indeed,  draw a great 2-sphere  $S_m\subset\mathbb S^3$ with the pole at the vertex $m$. It will cut the tetrahedron along a spherical triangle
which is called Lk($m$), the link of the vertex $m$ (see \cite{Luo3}). This triangle has, as one easily sees, the angles $\alpha_{ij}$, $\alpha_{ik}$ and $\alpha_{jk}$ and the side lengths $\alpha_{ij}^{(m)}$, $\alpha_{ik}^{(m)}$ and $\alpha_{jk}^{(m)}$, see Fig.~\ref{Fig5} (left). Therefore, we find:
\[
\big(\alpha_{ij}^{(m)},\alpha_{ik}^{(m)},\alpha_{jk}^{(m)}\big)=\Phi(\alpha_{ij},\alpha_{ik},\alpha_{jk}).
\]
In the above combinatorial interpretation, the dihedral angles $\alpha_{ij}$ are assigned to six 2-faces of the 4D cube, adjacent to one of its vertices. The four 3-faces (3D cubes) adjacent to this vertex, correspond combinatorially to the four spherical triangles Lk($m$). The result of application of our 3D system to these four 3D cubes can be written as $\alpha_{ij}^{(m)}$=$T_k\alpha_{ij}$.

At the second stage, we apply the cosine law to the four 
faces $\Delta(i,j,k)$ of the tetrahedron, see Fig.~\ref{Fig5} (right). In the above combinatorial interpretation, these four spherical triangles correspond to the four 3-faces of the cube adjacent to its opposite vertex. In particular, the side lengths $\ell_{ij}$ are assigned to the six 2-faces of the 4D cube, adjacent to this opposite vertex. We find:
\[
(\ell_{ij},\ell_{ik},\ell_{jk})=\Phi\big(\alpha_{ij}^{(k)},\alpha_{ik}^{(j)},\alpha_{jk}^{(i)}\big).
\]
This can be written as $\ell_{ij}=T_k\alpha_{ij}^{(k)}=T_k(T_m\alpha_{ij})$. Clearly, each edge is shared by two 2-faces, so we get two answers for each $\ell_{ij}$. Of course, they must coincide as definite geometric quantities in a given spherical tetrahedron. \qed
\medskip

The proof of Theorem \ref{Th: 4D} given above is completely free of computations and it is based on the existence of a spherical tetrahedron with given dihedral angles. We mention that very similar arguments were used in \cite{Luo2} in the opposite direction: the coincidence of the two answers for each $\ell_{ij}$ was demonstrated there by a direct computation and used to provide a geometric realization for a given angle structure.
\medskip

{\bf{Remark}}. In the combinatorial interpretation adopted in the present section, the map $\phi:\alpha\mapsto\ell$ is interpreted as solving a spherical triangle by its three planar angles, and combinatorially corresponds to one (directed) long diagonal of the cube. The inverse map $\phi^{-1}:\ell\mapsto\alpha$ would be interpreted as solving a spherical triangle by its three sides and would combinatorially correspond to the opposite orientation of this same long diagonal. It is easy to realize that other long diagonals of the cube would correspond to solving a spherical triangle by two sides and a planar angle enclosed by them, or by a side and two adjacent planar angles. The corresponding maps could be called {\em companion maps} to $\phi$. It is well known that companion maps for a multidimensionally consistent map with fields on 2-faces satisfy the so called {\em functional tetrahedron equation}, see, e.g., \cite[Sect. 6.16]{DDG}. We come to a conclusion that the cosine law for spherical triangles delivers an example of solutions of the functional tetrahedron equation. According to a private communication by R. Kashaev, he was aware of this example since many years.

\section{Spherical triangles and discrete time Euler top}
\label{sect: dEuler}

In the present section, it will be more convenient to change notation for the edges and angles of a spherical triangle to $(\ell_1,\ell_2,\ell_3)$ for $(\ell_{23},\ell_{13},\ell_{12})$, resp. to $(\alpha_1,\alpha_2,\alpha_3)$ for $(\alpha_{23},\alpha_{13},\alpha_{12})$. We study a dynamical system generated by iterations of the map $\phi$ from (\ref{c1 rat}), resp. its inverse map $\phi^{-1}$ from (\ref{c2 rat}), which now take the form
\beq
\phi: \;  y_{i}=  \frac{x_{i}+x_{j}x_{k}}
{\sqrt{1-x_{j}^2}\sqrt{1-x_{k}^2}},  \qquad \quad 
\phi^{-1}: \; x_{i}=  \frac{y_{i}-y_{j}y_{k}}
{ \sqrt{1-y_{j}^2}\sqrt{1-y_{k}^2}}.\label{c1 rat 1}
\eeq
Due to homogeneity, one can insert a small parameter $\epsilon$ into the map $\phi$ by a scaling transformation $x_i\mapsto \epsilon x_i$, $y_i\mapsto \epsilon y_i$ leading to
\beq\label{c1 rat eps}
y_{i}=  \frac{x_{i}+\epsilon x_{j}x_{k}}
{\sqrt{1-\epsilon^2x_j^2}\sqrt{1-\epsilon^2x_k^2}}.
\eeq
Geometrically, small values of $\epsilon$ in (\ref{c1 rat eps}) correspond to small values of $x_i=\cos\alpha_i$ in (\ref{c1 rat 1}), i.e., to spherical triangles with all three angles (and all three sides) being close to $\pi/2$. The form (\ref{c1 rat eps}) of the map $\phi$ makes it obvious that it can be considered as a time discretization of the famous {\em Euler top}, described by the following system of differential equations:
\beq\label{euler}
\dot x_i = x_j x_k.
\eeq
The latter system can be brought to the form $\dot x_i=c_ix_jx_k$, involving arbitrary real parameters $c_i$, by a simple scaling transformation $x_i\mapsto b_ix_i$, where, however, some of $b_i$ may be imaginary, depending on the signs of $c_i$.

It is well-known that system (\ref{euler}) is integrable in the Liouville-Arnold sense and can be solved in terms of elliptic functions, see for instance \cite{RSTS}. In particular, it is bi-Hamiltonian and admits integrals of motion $I_{ij}=x_i^2-x_j^2$, two of which are functionally independent.

We will show that discretization (\ref{c1 rat eps}) of the Euler top inherits integrability.

\begin{theorem}\label{Th 1}
The map $\phi$ admits integrals of motion
\beq\label{hy}
E_{ij}=\frac{1 -x_i^2}{1 - x_j^2},
\eeq
two of which are functionally independent. It has an invariant volume form given by
$$
\omega=\frac{dx_1\wedge dx_2\wedge dx_3}{\varphi(x)},
$$
where $\varphi(x)$ is any of the functions $(1 - x_j^2)(1-x_k^2)$ or $(1-x_i^2)^2$ (the quotient of any two such functions is an integral of motion).
\end{theorem}
\begin{proof}
The first statement (about integrals of motion) follows immediately from the sine law (\ref{si rat}). As for the second statement (about an invariant volume form), a direct computation leads to the following expressions:
\[
\frac{\partial y_i}{\partial x_i}=\frac{1}{\sqrt{1-x_j^2}\sqrt{1-x_k^2}},
\]
and
\[
\frac{\partial y_i}{\partial x_j}=\frac{x_k+x_ix_j}{\sqrt{(1-x_j^2)^3}\sqrt{1-x_k^2}}=
\frac{y_k}{\sqrt{1-x_j^2}\sqrt{1-x_k^2}}\,\sqrt{\frac{1-x_i^2}{1-x_j^2}}.
\]
The second factor in the last formula can be interpreted as a conjugation of the Jacobi matrix $\partial y/\partial x$ by the diagonal matrix with entries $\sqrt{1-x_i^2}$, and is therefore irrelevant for the computation of its determinant, which is thus equal to
\[
\det\left(\frac{\partial y}{\partial x}\right)=\frac{1}{(1-x_1^2)(1-x_2^2)(1-x_3^2)}
\left|\begin{array}{ccc} 1 & y_3 & y_2\\ y_3 & 1 & y_1 \\ y_2 & y_1 & 1\end{array}\right|
=\frac{d'}{(1-x_1^2)(1-x_2^2)(1-x_3^2)}.
\]
According to the sine law (\ref{si rat}), the latter expression is equal to
\[
\det\left(\frac{\partial y}{\partial x}\right)=\frac{(1-y_j^2)(1-y_k^2)}{(1-x_j^2)(1-x_k^2)},
\]
which finishes the proof.
\end{proof}

\begin{corollary}\label{biHam}
The map $\phi$ is Poisson with respect to Poisson brackets (bi-vector fields) obtained by contraction of the tri-vector field $\varphi(x)\partial_1\wedge \partial _2\wedge \partial_3$ with any of the one-forms $dE_{ij}(x)$. Explicitly, a family of compatible invariant Poisson brackets is given by
$$
\left\{x_i,x_j\right\} = C_i x_k (1-x_j^2)-C_j x_k(1-x_i^2),
$$
where $C_1$, $C_2$, $C_3$ are arbitrary constants. Thus, the map $\phi$ is bi-Hamiltonian and is completely integrable in the Liouville-Arnold sense.
\end{corollary}

A somewhat unpleasant feature of the map $\phi$ is its non-rational character, due to the square roots in the denominators. A striking and unexpected fact is that the second iterate of $\phi$ is a rational (actually, a birational) map, which appeared previously in the literature on a different occasion.

\begin{theorem}\label{th dET=F^2}
On the set $\tau\cap \phi^{-1}(\tau)$, where the second iterate of $\phi$ is defined, this second iterate $\phi^2:x\mapsto\t x$ is a birational map given by
\beq\label{de}
\widetilde{x}_i=\dfrac{x_i+2 x_jx_k+ x_i(-x_i^2+x_j^2+x_k^2)}
{1-(x_1^2+x_2^2+x_3^2) -2x_1x_2x_3}.
\eeq
These functions are the (unique) solutions of the following system of equations:
\beq\label{deuler}
\t x_i - x_i = \t x_j x_k + x_j \t x_k,
\eeq
which constitute the Hirota-Kimura discretization of the Euler top \cite{HK}.
\end{theorem}

{\it{Proof.}} We let $\phi(x)=y$, $\phi(y)=\t x$. According to formulas (\ref{c1 rat 1}), we have
\beq
 x_{i}=\frac{y_{i}-y_{j}y_{k}}{\sqrt{1-y_j^2}\sqrt{1-y_k^2}}, \qquad
 \t x_{i}=\frac{y_{i}+y_{j}y_{k}}{\sqrt{1-y_j^2}\sqrt{1-y_k^2}}. \label{he01 rat}
\eeq
Using (\ref{he01 rat}), we find, on one hand,
$$
\t x_i - x_i = \frac{2y_j y_k }{\sqrt{1-y_j^2}\sqrt{1-y_k^2}},
$$
and, on the other hand,
$$
\t x_j x_k + x_j \t x_k =\frac{(y_j+y_ky_i)(y_k-y_iy_j)+(y_j-y_ky_i)(y_k+y_iy_j)}{(1-y_i^2)\sqrt{1-y_j^2}\sqrt{1-y_k^2}}=
\frac{2 y_j y_k}{\sqrt{1-y_j^2}\sqrt{1-y_k^2}}.
$$
This proves equations (\ref{deuler}) for $\phi^2$. These equations are linear with respect to $\t{x}_1,\t{x}_2,\t{x}_3$, and thus they can be solved to give the explicit birational map (\ref{de}). \endpf

Actually, the integrability attributes of the map $\phi$ given in Theorem \ref{Th 1} and Corollary \ref{biHam} were found previously for the rational map $\phi^2$: integrals in \cite{HK} and the invariant volume form along with the bi-Hamiltonian structure in \cite{PS}. We refer also to \cite{HP,PPS1,PPS2} for our recent study on Hirota-Kimura-type discretizations.

While the map $\Phi$ relates the data within one spherical triangle, a natural interpretation of the map $\Phi^2$ is given by the following map between spherical triangles.

\begin{definition}
Let $\triangle$ be a spherical triangle with the angles $\alpha_i$ and the sides $\ell_i$. The {\em{switched triangle}} $\t \triangle=S(\triangle)$ is the spherical triangle  with inner angles
\beq\label{cond switch}
\t \alpha_i=\ell_i
\eeq
(whenever such a triangle exists). The switched triangle is defined up to a transformation from $O(3)$.
\end{definition}

Clearly, $S(\triangle)$ exists if and only if $\ell\in\CT\cap \Phi(\CT)$,
i.e., if and only if $\alpha\in\Phi^{-1}(\CT)\cap\CT$.
In this case, we have $\t\ell=\Phi^2(\alpha)$. Thus, the quantities $\t x_i=\cos\t \ell_i$ are related to $x_i=\cos\alpha_i$ by the Hirota-Kimura discrete time Euler top equations (\ref{de}) or (\ref{deuler}).
By abusing notations slightly, we write $(\t\alpha,\t\ell)=S(\alpha,\ell)$. By induction, we have:
$$
S^n(\alpha,\ell)=(\Phi^n(\alpha),\Phi^{n+1}(\alpha)).
$$
Therefore, $S^n(\triangle)$ is defined iff $\alpha\in \CT\cap \Phi^{-1}(\CT)\cap\ldots\cap\Phi^{-n}(\CT).$ For any $n$, the latter domain contains an open neighborhood of the point $(\pi/2,\pi/2,\pi/2)$.

When talking about the switched triangles to our colleagues, we have been often asked: ``But isn't this just the polar triangle?'' The answer is: ``Of course not, since the switch generates a non-trivial dynamical system, while the polarity is an involution''. Nevertheless, it turns out that these two transformations of spherical triangles are closely related. We denote the angles and the sides of the polar triangle $\triangle^*=P(\Delta)$ by $(\alpha^*, \ell^*)=P(\alpha,\ell)$. They are given by
\beq\label{cond polar}
\alpha_i^*=\pi-\ell_i, \qquad \ell_i^*=\pi-\alpha_i.
\eeq
Clearly, in the latter interpretation, $P$ maps $\CT\times\CT^*$ to itself. Comparing formulas (\ref{cond switch}) and (\ref{cond polar}), we are led to the following definition.

\begin{definition}
Let $\triangle$ be a spherical triangle with the angles $\alpha_i$ and the sides $\ell_i$. The {\em{side flip}} $\h \triangle=F(\triangle)$ is the spherical triangle with sides
\beq\label{cond flip}
\h \ell_i=\pi-\ell_i
\eeq
(whenever such a triangle exists). The side flip $F(\triangle)$ is defined up to a transformation from $O(3)$. The map $F$ is an involution (when considered on equivalence classes modulo $O(3)$).
\end{definition}

An immediate consequence of the above definition is the following claim.

\begin{theorem}\label{prop S=GP}
The switch of spherical triangles, considered as a map on equivalence classes modulo $O(3)$, can be represented as a composition of two involutions, namely as the side flip followed by the polarity transformation:
$
S=P\circ F.
$
\end{theorem}
The angles $\h \alpha_i$ of the side flipped triangle $\h\triangle=F(\triangle)$ are given by $\h\alpha=\pi-\Phi^2(\alpha)$. The quantities $x_i=\cos \alpha_i$ and $\h x_i = \cos \h \alpha_i$ satisfy
\beq\label{Jonas impl}
\h x_i+x_i = \h x_j x_k + x_j \h x_k,
\eeq
or, explicitly,
\beq\label{Jonas expl}
\h{x}_i=-\dfrac{x_i+2  x_jx_k+x_i(-x_i^2+x_j^2+x_k^2)}
{1-(x_1^2+x_2^2+x_3^2)-2x_1x_2x_3}.
\eeq
The map $f:x\mapsto\h x$ defined by (\ref{Jonas expl}) is a birational involution.

It is very easy to find integrals of this map. Indeed, from the defining relation (\ref{cond flip}) yielding $\cos\h \ell_i=-\cos\ell_i$ and the cosine law (\ref{c1}) there follows:
$$
\frac{\cos \h \alpha_i + \cos\h \alpha_j \cos \h\alpha_k}{\sin \h\alpha_j \sin\h \alpha_k}=-
\frac{\cos \alpha_i + \cos \alpha_j \cos \alpha_k}{\sin \alpha_j \sin \alpha_k}.
$$
As a consequence, one gets three functionally independent conserved quantities of the map (\ref{Jonas expl}):
$$
\cos^2\ell_i=\frac{(x_i+x_j x_k)^2}{(1-x_j^2)(1-x_k^2)} \qquad\Leftrightarrow\qquad \sin^2\ell_i=\frac{d}{(1-x_j^2)(1-x_k^2)},
$$
with the quantity $d$ defined by (\ref{po rat 1}). Note that the quotients
$
E_{ij}= \sin^2\ell_i / \sin^2\ell_j
$
depend only on $x_i^2$  rather than on $x_i$, and are therefore integrals of motion for the map $\phi^2=i\circ f$ (these are nothing but functions (\ref{hy})).

Map (\ref{Jonas expl}) appeared in an old paper \cite{J} by H.~Jonas. Actually, he considered a closely related {\em angle flip} $\triangle\mapsto \overline \triangle$ of spherical triangles, defined by the relation
$
\overline \alpha_i=\pi-\alpha_i,
$
and proved that the cosines of the side lengths of the spherical triangles $\triangle, \overline \triangle$ are related by
\begin{equation}\label{Jonas}
\overline x_i+x_i + \overline x_j x_k + x_j \overline x_k=0.
\end{equation}
The latter map is conjugated to (\ref{Jonas impl}) via $i:x\mapsto -x$. Jonas' results include integrals of the map (\ref{Jonas}) and its solution in terms of elliptic functions.

\section{Spherical tetrahedra and discretization of two linearly coupled Euler tops}
\label{sect: dcoupled Euler}

We now consider the algebraic form of the cosine law for spherical tetrahedra as a dynamical system. Recall that it is given by
\beq
y_{ij}=\frac{g_{ij}}{\sqrt{g_{ii}g_{jj}}},
\eeq
with the expressions for the polynomials $g_{ij}$ given in (\ref{sph tetr c1}--\ref{sph tetr c2}). This map, which we denote by $\psi$,  can be considered as a time discretization of the following system of ordinary differential equations:
\beq\label{two tops coupled}
\dot{x}_{ij}=x_{ik}x_{jk}+x_{im}x_{jm}.
\eeq
Here, recall, $(i,j,k,m)$ stands for an arbitrary permutation of $(1,2,3,4)$. This six-dimensional system turns out to consist of two linearly coupled copies of the Euler top. Namely, each group of three variables $p_{ij}=x_{ij}+x_{km}$ and $q_{ij}=x_{ij}-x_{km}$ with $ij=12,13,23$, satisfies its own copy of the Euler top equations:
\beq\label{two tops uncoupled}
\dot{p}_{ij}=p_{ik}p_{jk}, \qquad \dot{q}_{ij}=q_{ik}q_{jk}.
\eeq
As a consequence, equations (\ref{two tops coupled}) are integrable in the Liouville-Arnold sense, with four independent integrals of motion which can be chosen as
\begin{eqnarray*}
p_{12}^2-p_{13}^2=(x_{12}+x_{34})^2-(x_{13}+x_{24})^2, & \; & p_{13}^2-p_{23}^2=(x_{13}+x_{24})^2-(x_{23}+x_{14})^2, \\
q_{12}^2-q_{13}^2=(x_{12}-x_{34})^2-(x_{13}-x_{24})^2, & \; & q_{13}^2-q_{23}^2=(x_{13}-x_{24})^2-(x_{23}-x_{14})^2,
\end{eqnarray*}
or, alternatively, as
\begin{eqnarray*}
x_{12}^2+x_{34}^2-x_{13}^2-x_{24}^2, & \; & x_{13}^2+x_{24}^2-x_{23}^2-x_{14}^2, \\
x_{12}x_{34}-x_{13}x_{24}, & \; & x_{13}x_{24}-x_{23}x_{14}.
\end{eqnarray*}

\begin{theorem}\label{Th 2}
The map $\psi$ admits four independent integrals of motion
$$
\frac{(1-x_{13}^2)(1-x_{24}^2)}{(1-x_{12}^2)(1-x_{34}^2)},\qquad \frac{(1-x_{23}^2)(1-x_{14}^2)}{(1-x_{12}^2)(1-x_{34}^2)},
$$
$$
\frac{x_{12}x_{34}-x_{13}x_{24}}{\sqrt{(1-x_{12}^2)(1-x_{34}^2)}},\qquad \frac{x_{13}x_{24}-x_{23}x_{14}}{\sqrt{(1-x_{12}^2)(1-x_{34}^2)}}.
$$
It has an invariant volume form given by
$$
\omega=\frac{dx_{12}\wedge dx_{13}\wedge dx_{23}\wedge dx_{14}\wedge dx_{24}\wedge dx_{34}}{\varphi(x)},
$$
where $\varphi(x)$ is any of the functions $(1 - x_{ij}^2)^{5/2}(1 - x_{km}^2)^{5/2}$.
\end{theorem}

\begin{proof} The claim about integrals of motion follows directly from the sine laws (\ref{67}--\ref{68}). Let us turn to the proof of the claim about the invariant volume form. Our aim is to compute the determinant of the Jacobi matrix $(\partial y_{ij}/\partial x_{i'j'})$. We use the following ordering of the index pairs: 12, 13, 23, 14, 24, 34. The entries of the Jacobi matrix are given by the following expressions:
\begin{eqnarray}
\frac{\partial y_{ij}}{\partial x_{ij}} & = & \frac{1-x_{km}^2}{\sqrt{g_{ii}g_{jj}}}, \label{eq: ijij}\\
\frac{\partial y_{ij}}{\partial x_{ik}} & = &
\frac{1-x_{km}^2}{\sqrt{g_{ii}g_{jj}}}\, y_{jk}\, \sqrt{\frac{g_{kk}}{g_{jj}}}, \label{eq: ijik}\\
\frac{\partial y_{ij}}{\partial x_{km}} & = &
\frac{1-x_{km}^2}{\sqrt{g_{ii}g_{jj}}}\,
p_{km}^{ij}\, \sqrt{\frac{g_{kk}g_{mm}}{g_{ii}g_{jj}}}, \label{eq: ijkm}
\end{eqnarray}
where
\beq \nn
p_{km}^{ij}=\frac{y_{ik}y_{jm}+y_{im}y_{jk}-y_{ij}y_{ik}y_{im}-y_{ij}y_{jk}y_{jm}}{1-y_{ij}^2}.
\eeq
Indeed, the first two expressions, (\ref{eq: ijij}--\ref{eq: ijik}), are easily obtained by a direct computation, while the third one, (\ref{eq: ijkm}), is a result of \cite{Luo3}. The last factors in the above expressions can be interpreted as a conjugation of the Jacobi matrix $(\partial y_{ij}/\partial x_{i'j'})$ by a diagonal matrix
\[
{\rm diag}\left(1,\sqrt{\frac{g_{33}}{g_{22}}},\sqrt{\frac{g_{33}}{g_{11}}},\sqrt{\frac{g_{44}}{g_{22}}},
\sqrt{\frac{g_{44}}{g_{11}}},\sqrt{\frac{g_{33}g_{44}}{g_{11}g_{22}}}\right),
\]
and are therefore inessential for the computation of the determinant. We come to the following result:
$$
\det\left(\frac{\partial y}{\partial x}\right)={\displaystyle{
\frac{\prod_{k<m} (1-x_{km}^2)}{(g_{11}g_{22}g_{33}g_{44})^{3/2}}
}}
\left|\begin{array}{cccccc}
1 & y_{23} & y_{13} & y_{24} & y_{14} & p_{12}^{34} \\
y_{23} & 1 & y_{12} & y_{34} & p_{13}^{24} & y_{14} \\
y_{13} & y_{12} & 1 & p_{23}^{14} & y_{34} & y_{24} \\
y_{24} & y_{34} & p_{14}^{23} & 1 & y_{12} & y_{13} \\
y_{14} & p_{24}^{13} & y_{34} & y_{12} & 1 & y_{23} \\
p_{34}^{12} & y_{14} & y_{24} & y_{13} & y_{23} & 1
\end{array}\right|.
$$
A direct computation with a symbolic manipulation system like Maple leads to the following factorized expression for the determinant on the right-hand side:
$$
\left|\begin{array}{cccccc}
1 & y_{23} & y_{13} & y_{24} & y_{14} & p_{12}^{34} \\
y_{23} & 1 & y_{12} & y_{34} & p_{13}^{24} & y_{14} \\
y_{13} & y_{12} & 1 & p_{23}^{14} & y_{34} & y_{24} \\
y_{24} & y_{34} & p_{14}^{23} & 1 & y_{12} & y_{13} \\
y_{14} & p_{24}^{13} & y_{34} & y_{12} & 1 & y_{23} \\
p_{34}^{12} & y_{14} & y_{24} & y_{13} & y_{23} & 1
\end{array}\right|=\frac{g'_{11}g'_{22}g'_{33}g'_{44}}{\prod_{i<j} (1-y_{ij}^2)} \, d'.
$$
Therefore,
$$
\det\left(\frac{\partial y}{\partial x}\right)=
\frac{\prod_{k<m} (1-x_{km}^2)}{\prod_{i<j} (1-y_{ij}^2)}
\,\frac{(g'_{11}g'_{22}g'_{33}g'_{44})^{3/2}}{(g_{11}g_{22}g_{33}g_{44})^{3/2}}
\,\frac{d'}{\gamma'}.
$$
At the next step, we use the following result (proved below):
\begin{equation} 
\label{eq:ggs}
\frac{(g'_{11}g'_{22}g'_{33}g'_{44})^{3/2}}{(g_{11}g_{22}g_{33}g_{44})^{3/2}}=
\frac{\prod_{i<j} (1-y_{ij}^2)^2}{\prod_{k<m} (1-x_{km}^2)^2}.
\end{equation}
As a consequence, we find with the help of (\ref{67}):
$$
\det\left(\frac{\partial y}{\partial x}\right)=
\frac{\prod_{i<j} (1-y_{ij}^2)}{\prod_{k<m} (1-x_{km}^2)}\,\frac{d'}{\gamma'}=\left(\frac{\gamma'}{d'}\right)^6\,
\frac{d'}{\gamma'}=\left(\frac{\gamma'}{d'}\right)^5,
$$
and one more reference to relations (\ref{67}) proves the claim about the invariant volume form.

To finish the proof we need to prove formula (\ref{eq:ggs}). Comparing expression (\ref{po rat 1}) with (\ref{sph tetr c2}), we see that $g_{ii}$ is the analogue of the quantity $d$ for the spherical triangle with the angles $\alpha_{jk}$, $\alpha_{jm}$, $\alpha_{km}$, i.e., for the triangle Lk($i$). Then, from (\ref{eq: sine law prod 1}), we obtain:
\[
g_{ii}^{1/2}=\sin\alpha_{jk}^{(i)}\sin\alpha_{jm}\sin\alpha_{km}.
\]
Multiplying all three such expressions for Lk($i$), we find:
\[
g_{ii}^{3/2}=\sin\alpha_{jk}^{(i)}\sin\alpha_{jm}^{(i)}\sin\alpha_{km}^{(i)}
\sin^2\alpha_{jk}\sin^2\alpha_{jm}\sin^2\alpha_{km}.
\]
Similarly, $g'_{ii}$ is the analogue of the quantity $d'$ for the spherical triangle with the sides $\ell_{jk}$, $\ell_{jm}$, $\ell_{km}$, i.e., for the triangle $\Delta(j,k,m)$. From (\ref{eq: sine law prod 1}), we obtain:
\[
(g'_{ii})^{1/2}=\sin\alpha_{jk}^{(m)}\sin\ell_{jm}\sin\ell_{km}.
\]
Multiplying all three such expressions for $\Delta(j,k,m)$, we find:
\[
(g'_{ii})^{3/2}=\sin\alpha_{jk}^{(m)}\sin\alpha_{jm}^{(k)}\sin\alpha_{km}^{(j)}
\sin^2\ell_{jk}\sin^2\ell_{jm}\sin^2\ell_{km}.
\]
As a consequence,
\begin{eqnarray*}
(g_{11}g_{22}g_{33}g_{44})^{3/2} & = & \left( \prod_{i=1}^4 \prod_{j<k}\sin\alpha_{jk}^{(i)}\right) \prod_{j<k}\sin^4\alpha_{jk},\\
(g'_{11}g'_{22}g'_{33}g'_{44})^{3/2} & = & \left( \prod_{m=1}^4 \prod_{j<k}\sin\alpha_{jk}^{(m)} \right) \prod_{j<k}\sin^4\ell_{jk}.
\end{eqnarray*}
Indeed, each dihedral angle $\alpha_{jk}$ participates as a planar angle of two triangles Lk($i$) and Lk($m$), and each side $\ell_{jk}$ is shared by two triangles $\Delta(j,k,m)$ and $\Delta(i,j,k)$. This  proves (\ref{eq:ggs}). 
\end{proof}

{\bf Remark.} According to the famous Schl\"afli formula (see, e.g., \cite{V, Luo3}), if $V=V(\alpha)$ is the volume of the spherical tetrahedron with the dihedral angles $\alpha_{ij}$, then $dV=\frac{1}{2}\sum_{ij}\ell_{km}d\alpha_{ij}$. Thus, our proof of Theorem \ref{Th 2} actually leads to an evaluation of the determinant of the Hesse matrix of the volume function, $\det(\partial^2 V/\partial \alpha_{ij}\partial\alpha_{i'j'})$. We are not aware of such an evaluation in the literature.
\medskip

{\bf Remark.} The reduction of the system of ordinary differential equations (\ref{two tops coupled}) to two copies of the Euler top (\ref{two tops uncoupled}) yields that the original system admits two invariant 3-forms, $dp_{12}\wedge dp_{13}\wedge dp_{23}$ and $dq_{12}\wedge dq_{13}\wedge dq_{23}$. We conjecture that the discretization $\psi$ also admits an invariant 3-form. This property is stronger than existence of an invariant volume form established in Theorem \ref{Th 2}. If true, this 3-form (or, better, the corresponding 3-vector field) would yield, by contraction with gradients of integrals of motion, a family of compatible Poisson structures for $\psi$.

\section*{Acknowledgment}
The authors are partly supported by DFG (Deutsche Forschungsgemeinschaft) in the frame of Sonderforschungsbereich/Transregio 109 ``Discretization in Geometry and Dynamics''.


\end{document}